\documentclass[format=acmsmall, review=false]{acmart}

\usepackage{verbatim}
\usepackage{amsmath}

\usepackage{rotating}
\usepackage{algorithm}
\usepackage[noend]{algpseudocode}

\usepackage{tabularx}
\usepackage{subfigure}
\usepackage{graphicx}
\usepackage{url}
\usepackage{bigfoot}
\usepackage{multirow}

\usepackage{acm-ec-17}

\usepackage{silence}
\begin{document}

\title[Bayesian Opponent Exploitation in Imperfect-Information Games]{\Large Bayesian Opponent Exploitation in Imperfect-Information Games}  
\author{Sam Ganzfried}
\affiliation{%
  \institution{Florida International University and Ganzfried Research}
  \department{School of Computing and Information Sciences}
  \streetaddress{11200 S.W. 8th Street}
  \city{Miami}
  \state{FL}
  \postcode{33199}
  \country{USA}}
  \author{Qingyun Sun}
\affiliation{%
  \institution{Stanford University}
  \department{School of Mathematics}
  \streetaddress{10 Comstock Circle}
  \city{Stanford}
  \state{CA}
  \postcode{94305}
  \country{USA}}

\begin{abstract}
Two fundamental problems in computational game theory are computing a Nash equilibrium and learning to exploit opponents given observations of their play (opponent exploitation). The latter is perhaps even more important than the former: Nash equilibrium does not have a compelling theoretical justification in game classes other than two-player zero-sum, and for all games one can potentially do better by exploiting perceived weaknesses of the opponent than by following a static equilibrium strategy throughout the match. The natural setting for opponent exploitation is the Bayesian setting where we have a prior model that is integrated with observations to create a posterior opponent model that we respond to. The most natural, and a well-studied prior distribution is the Dirichlet distribution. An exact polynomial-time algorithm is known for best-responding to the posterior distribution for an opponent assuming a Dirichlet prior with multinomial sampling in normal-form games; however, for imperfect-information games the best known algorithm is based on approximating an infinite integral without theoretical guarantees. We present the first exact algorithm for a natural class of imperfect-information games. We demonstrate that our algorithm runs quickly in practice and outperforms the best prior approaches. We also present an algorithm for the uniform prior setting.
\end{abstract}

\maketitle

\section{Introduction}
\label{se:intro}

Imagine you are playing a game repeatedly against one or more opponents. What algorithm should you use to maximize your performance? The classic ``solution concept'' in game theory is the Nash equilibrium. In a Nash equilibrium $\sigma$, each player is simultaneously maximizing his payoff assuming the opponents all follow their components of $\sigma$. So should we just find a Nash equilibrium strategy for ourselves and play it in all the game iterations?

Unfortunately, there are some complications. First, there can exist many Nash equilibria, and if the opponents are not following the same one that we have found (or are not following one at all), then our strategy would have no performance guarantees. Second, finding a Nash equilibrium is challenging computationally: it is PPAD-hard and is widely conjectured that no polynomial-time algorithms exist~\cite{Chen06:Settling}. These challenges apply to both extensive-form games (of both perfect and imperfect information) and strategic-form games, for games with more than two players and non-zero-sum games. While a particular Nash equilibrium may happen to perform well in practice,\footnote{An agent for 3-player limit Texas hold 'em computed by the counterfactual regret minimization algorithm (which converges to Nash equilibrium in certain games) performed well in practice despite a lack of theoretical justification~\cite{Gibson14:Regret}.} there is no theoretically compelling justification for why computing one and playing it repeatedly is a good approach. Two-player zero-sum games do not face these challenges: there exist polynomial-time algorithms for computing an equilibrium~\cite{Koller94:Fast}, and there exists a game value that is guaranteed in expectation in the worst case by all equilibrium strategies regardless of the strategy played by the opponent (and this value is the best worst-case guaranteed payoff for any of our strategies). However, even for this game class it would be desirable to deviate from equilibrium in order to learn and exploit perceived weaknesses of the opponent; for instance, if the opponent has played Rock in each of the first thousand iterations of rock-paper-scissors, it seems desirable to put additional weight on paper beyond the equilibrium value of $\frac{1}{3}.$

Thus, learning to exploit opponents' weaknesses is desirable in all game classes. One approach would be to construct an opponent model consisting of a single mixed strategy that we believe the opponent is playing given our observations of his play and a prior distribution (perhaps computed from a database of historical play). This approach has been successfully applied to exploit weak agents in limit Texas hold 'em poker, a large imperfect-information game~\cite{Ganzfried11:Game}.\footnote{This approach used an approximate Nash equilibrium strategy as the prior and is applicable even when historical data is not available, though if additional data were available a more informed prior that capitalizes on the data would be preferable.} A drawback is that it is potentially not robust. It is very unlikely that the opponent's strategy matches this point estimate exactly, and we could perform poorly if our model is incorrect. A more robust approach, which is the natural one to use in this setting, is to use a Bayesian model, where the prior and posterior are full distributions over mixed strategies of the opponent, not single mixed strategies. A natural prior distribution, which has been studied and applied in this context, is the Dirichlet distribution. The pdf of the Dirichlet distribution is the belief that the probabilities of $K$ rival events are $x_i$ given that each event has been observed $\alpha_i -1$ times: $f(x,\alpha) = \frac{1}{B(\alpha)}\prod x^{\alpha_i-1}_i.$\footnote{\label{fo:Dirichlet}$B(\alpha)$  is the beta function $B(\alpha) = \frac{\prod{\Gamma(\alpha_i)}}{\Gamma \left(\sum_i \alpha_i\right)}$, where $\Gamma(n) = (n-1)!$ is the gamma function.} Some notable properties are that the mean is $E[X_i] = \frac{\alpha_i}{\sum_k \alpha_k}$ and that, assuming multinomial sampling, the posterior after including new observations is also Dirichlet, with parameters updated based on the new observations. 

Prior work has presented an efficient algorithm for optimally exploiting an opponent in normal-form games in the Bayesian setting with a Dirichlet prior~\cite{Fudenberg98:Theory}. The algorithm is essentially the fictitious play rule~\cite{Brown51:Iterative}. Given prior counts $\alpha_i$ for each opponent action, the algorithm increments the counter for an action by one each time it is observed, and then best responds to a model for the opponent where he plays each strategy in proportion to the counters. This algorithm would also extend directly to sequential extensive-form games of perfect information, where we maintain independent counters at each of the opponent's decision nodes; this would also work for games of imperfect information where the opponent's private information is observed after each round of play (so that we would know exactly what information set he took the observed action from). For all of these game classes the algorithm would apply to both zero and general-sum games, for any number of players. However, it would not apply to imperfect-information games where the opponent's private information is not observed after play. 

An algorithm exists for approximating a Bayesian best response in imperfect-information games, which uses importance sampling to approximate an infinite integral. This algorithm has been successfully applied to limit Texas hold 'em poker~\cite{Southey05:Bayes}.\footnote{In addition to Bayesian Best Response, the paper also considers heuristic approaches for approximating several other response functions: Max A Posteriori Response and Thompson's Response.} However, it is only a heuristic approach with no guarantees. The authors state,
\begin{quote}
``Computing the integral over opponent strategies depends on the form of the prior but is difficult in any event. For Dirichlet priors, it is possible to compute the posterior exactly but the calculation is expensive except for small games with relatively few observations. This makes the exact BBR an ideal goal rather than a practical approach. For real play, we must consider approximations to BBR.''
\end{quote}
\normalsize
However, we see no justification for the claim that it is possible to compute the posterior exactly in prior work, and there could easily be no closed-form solution. In this paper we present a solution for this problem, leading to the first exact optimal algorithm for performing Bayesian opponent exploitation in imperfect-information games. While the claim is correct that the computation is expensive for large games, we show that in a small (yet realistic) game it outperforms all prior approaches. Furthermore, we show that the computation can run extremely quickly even for large number of observations (though it can run into numerical instability), contradicting the second claim. We also present general theory, and an algorithm for another natural prior distribution (uniform distribution over a polyhedron). 

\section{Meta-algorithm}
\label{se:meta-algorithm}
The problem of developing efficient algorithms for optimizing against a posterior distribution, which is a probability distribution over mixed strategies for the opponent (which are themselves distributions over pure strategies) seems daunting. We need to be able to compactly represent the posterior distribution and efficiently compute a best response to it. Fortunately, we show that our payoff of playing any strategy $\sigma_i$ against a probability distribution over mixed strategies for the opponent equals our payoff of playing $\sigma_i$ against the mean of the distribution. Thus, we need only represent and respond to the single strategy that is the mean of the distribution, and not to the full distribution. While this result was likely known previously, we have not seen it stated explicitly, and it is important enough to be highlighted so that it is on the radar of the AI community. 
 
Suppose the opponent is playing mixed strategy $\sigma_{-i}$ where $\sigma_{-i}(s_{-j})$ is the probability that he plays pure strategy $s_{-j} \in S_{-j}$. By definition of expected utility,
$u_i(\sigma_i, \sigma_{-i}) = \sum_{s_{-j} \in S_{-j}} \sigma_{-i}(s_{-j}) u_i(\sigma_i, s_{-j}).$

We can generalize this naturally to the case where the opponent is playing according to a probability distribution with pdf $f_{-i}$ over mixed strategies:
$$ u_i(\sigma_i, f_{-i}) = \int _{\sigma_{-i} \in \Sigma_{-i}} \left[ f_{-i}(\sigma_{-i}) \cdot u_i(\sigma_i, \sigma_{-i}) \right].$$
\normalsize
Let $\overline{f_{-i}}$ denote the mean of $f_{-i}$. That is, $\overline{f_{-i}}$ is the mixed strategy that selects $s_{-j}$ with probability
$$\int _{\sigma_{-i} \in \Sigma_{-i}} \left[ \sigma_{-i}(s_{-j}) \cdot f_{-i}(\sigma_{-i})   \right] .$$
\normalsize
Then we have the following:

\small
\begin{theorem}
$$u_i(\sigma_i, \overline{f_{-i}}) = u_i(\sigma_i, f_{-i}).$$ That is, the payoff against the mean of a strategy distribution equals the payoff against the full distribution.
\label{th:mean}
\end{theorem}
\normalsize
\begin{proof}
\begin{eqnarray*}
&& u_i(\sigma_i, \overline{f_{-i}}) \\
&= &\sum_{s_{-j} \in S_{-j}} \left[ u_i(\sigma_i, s_{-j}) \int _{\sigma_{-i} \in \Sigma_{-i}} \left[ \sigma_{-i}(s_{-j}) \cdot f_{-i}(\sigma_{-i})   \right] \right] \\
& = & \sum_{s_{-j} \in S_{-j}} \left[ \int _{\sigma_{-i} \in \Sigma_{-i}} \left[ u_i(\sigma_i, s_{-j}) \cdot \sigma_{-i}(s_{-j}) \cdot f_{-i}(\sigma_{-i}) \right] \right] \\
& = & \int_{\sigma_{-i} \in \Sigma_{-i}} \left[ \sum_{j \in S_{-j}} \left[ u_i(\sigma_i, s_{-j}) \cdot \sigma_{-i}(s_{-j}) \cdot f_{-i}(\sigma_{-i}) \right] \right] \\
& = & \int_{\sigma_{-i} \in \Sigma_{-i}} \left[ u_i(\sigma_i, \sigma_{-i}) \cdot f_{-i}(\sigma_{-i}) \right] \\
& = & u_i(\sigma_i, f_{-i})
\end{eqnarray*}
\end{proof}
\normalsize

Theorem~\ref{th:mean} applies to both normal and extensive-form games (with perfect or imperfect information), for any number of players ($\sigma_{-i}$ could be a joint strategy profile for all opposing agents).

Now suppose the opponent is playing according a prior distribution $p(\sigma_{-i})$, and let $p(\sigma_{-i} | x)$ denote the posterior probability given observations $x$. Let $\overline{p(\sigma_{-i} | x)}$ denote the mean of $p(\sigma_{-i} | x)$. As an immediate consequence of Theorem~\ref{th:mean}, we have the following corollary. 
\begin{corollary}
$u_i(\sigma_i, \overline{p(\sigma_{-i} | x)}) = u_i(\sigma_i, p(\sigma_{-i} | x))$. 
\label{co:posterior-mean}
\end{corollary}

Corollary~\ref{co:posterior-mean} implies the meta-procedure for optimizing performance against an opponent who uses $p$ given by Algorithm~\ref{al:meta}. 

\begin{algorithm}[!ht]
\caption{\normalsize Meta-algorithm for Bayesian opponent exploitation \normalsize}
\label{al:meta} 
\normalsize
\textbf{Inputs}: Prior distribution $p_0$, response functions $r_t$ for $0 \leq t \leq T$
\begin{algorithmic}
\State $M_0 \gets \overline{p_0(\sigma_{-i})}$
\State $R_0 \gets r_0(M_0)$
\State Play according to $R_0$
\For {$t = 1$ to $T$}
\State $x_t \gets $ observations of opponent's play at time step $t$
\State $p_t \gets$ posterior distribution of opponent's strategy given prior $p_{t-1}$ and observations $x_t$ 
\State $M_t \gets$ mean of $p_t$
\State $R_t \gets r_t(M_t)$
\State Play according to $R_t$
\EndFor
\end{algorithmic}
\normalsize
\end{algorithm}

There are several challenges for applying Algorithm~\ref{al:meta}. First, it assumes that we can compactly represent the prior and posterior distributions $p_t$, which have infinite domain (the set of opponents' mixed strategy profiles). Second, it requires a procedure to efficiently compute the posterior distributions given the prior and the observations, which requires updating potentially infinitely many strategies. Third, it requires an efficient procedure to compute the mean of $p_t$. And fourth, it requires that the full posterior distribution from one round be compactly represented to be used as the prior in the next round. We can address the fourth challenge by using a modified update step: 
$$p_t \gets \mbox{ posterior distribution of opponent's strategy} \mbox{ given prior } p_0 \mbox{ and observations } \\x_1,\ldots,x_t.$$ 
We will be using this new rule in our main algorithm.

The response functions $r_t$ could be standard best response, for which linear-time algorithms exist in games of imperfect information (and a recent approach has enabled efficient computation in extremely large games~\cite{Johanson11:Accelerating}). It could also be a more robust response, e.g., one that places a limit on the exploitability of our own strategy, perhaps one that varies over time based on our performance (or a lower-variance estimator of it)~\cite{Johanson07:Computing,Johanson09:Data,Ganzfried15:Safe}. In particular, the restricted Nash response has been demonstrated to outperform best response against agents in limit Texas hold 'em whose actual strategy may differ substantially from the exact model~\cite{Johanson07:Computing}.

\section{Robustness of the approach}
\label{se:robustness}
It has been pointed out that, empirically, the approach described is not robust: if we play a full best response to a point estimate of the opponent's strategy we can have very high exploitability ourselves, and could perform very poorly if in fact we are wrong about our model~\cite{Johanson07:Computing}. This could happen for several reasons. Our modeling algorithm could be incorrect: it could make an incorrect assumption about the prior and form of the opponent's distribution. This could happen for several reasons. One reason is that the opponent could actually be changing his strategy over time (possibly either by improving his own play or by adapting to our play), in which case a model that assumes a static opponent could be predicting a strategy that the opponent is no longer using. The opponent could also have modified his play strategically in an attempt to deceive us by playing one way initially and then counter-exploiting us as we attempt to exploit the model we have formed from his initial strategy (e.g., the opponent initially starts off playing extremely conservatively, then switches to a more aggressive style as he suspects we will start to exploit his extreme conservatism). His initial strategy need not arise from deception: it is also possible that simply due to chance events (either due to his own randomization in his strategy or due to the states of private information selected by chance) the opponent has appeared to be playing in a certain way (e.g., very conservatively), and as he becomes aware of this conservative ``image,'' naturally it occurs to him to modify his play by becoming more aggressive. 

A second reason that we could be wrong in our opponent model other than our modeling algorithm incorrectly modeling the opponents' dynamic approach is that our observations of his play are very noisy (due to both randomization in the opponent's strategy and to the private information selected by chance), particularly over a small sample.  Even if our approach is correct and the opponent is in fact playing a static strategy according to the distribution assumed by the modeling algorithm, it is very unlikely that our actual perception of his strategy is precisely correct.

A third reason, of course, is that the opponent may be following a static strategy that does not exactly conform to our model for the prior and/or sampling method used to generate the posterior.

We would like an approach that is robust in the event that our model of the opponent's strategy is incorrect, whichever the cause may be. Prior work has considered a model where the opponent plays according to a model $x_{-i}$ with probability $p$ and with probability $1-p$ plays a nemesis to our strategy~\cite{Johanson07:Computing}. For carefully selected values of $p$ (typically 0.95 or 0.99), they show that this can achieve a relatively high level of exploitation (similar to a full best response) with a significantly smaller worst-case exploitability. We note that, as described in Section~\ref{se:meta-algorithm}, Algorithm~\ref{al:meta} can be integrated with any response function, not necessarily a full best response, and so $r_t$ could be selected to be the Restricted Nash Response from prior work~\cite{Johanson07:Computing}. However, it seems excessively conservative to give the opponent credit for playing a full nemesis to our strategy; if we are relatively confident in our opponent model, then a more reasonable robustness criterion would be to explore performance as we allow the opponent's strategy to differ by a small amount from the predicted strategy (i.e., the opponent is playing a strategy that is very close to our model, and not necessarily putting weight on a full nemesis to our strategy).

Suppose we believe the opponent is playing $x_{-i}$, while he is actually playing $x'_{-i}$. Let $M$ be the maximum absolute value of a utility to player $i$, and let $N$ be the maximum number of actions available to a player. Let $\epsilon > 0$ be arbitrary. Then, if $|x_{-i}(j) - x'_{-i}(j)| < \delta$ for all $j$, where $\delta = \frac{\epsilon}{MN}$, 

\begin{eqnarray*}
&& |u_i(\sigma^*,x_{-i}) - u_i(\sigma^*,x'_{-i})| = \left|\sum_j (x_{-i}(j) - x'_{-i}(j))u_i(\sigma^*,s_{-j}) \right| \\
& <= &\sum_j \left| \left( x_{-i}(j) - x'_{-i}(j) \right) u_i(\sigma^*,s_{-j}) \right| <= \sum_j \left( \left|x_{-i}(j) - x'_{-i}(j)\right| \cdot \left|u_i(\sigma^*,s_{-j})\right| \right) \\
& <= &\sum_j  \left( |x_{-i}(j) - x'_{-i}(j)| \cdot M \right) < M \sum_j \delta <= MN \delta = MN \cdot \frac{\epsilon}{MN} = \epsilon\\
\end{eqnarray*}
This same analysis can be applied directly to show that our payoff is continuous in the opponent's strategy for many popular distance functions (i.e., for any distance function where one strategy can get arbitrarily close to another as the components get arbitrarily close). For instance this would apply to L1, L2, and earth mover's distance, which have been applied previously to compute distances been strategies within opponent exploitation algorithms~\cite{Ganzfried11:Game}. Thus, if we are slightly off in our model of the opponent's strategy, even if we are doing a full best response we will do only slightly worse.

\section{Exploitation algorithm for Dirichlet prior}
\label{se:main-algorithm}
As described in Section~\ref{se:intro} the Dirichlet distribution is the conjugate prior for the multinomial distribution, and therefore the posterior is also a Dirichlet distribution, with the parameters $\alpha_i$ updated to reflect the new observations. Thus, the mean of the posterior can be computed efficiently by computing the strategy for the opponent in which he plays each strategy in proportion to the updated weight, and Algorithm~\ref{al:meta} yields an exact efficient algorithm for computing the Bayesian best response in normal-form games with a Dirichlet prior. However, the algorithm does not apply to games of imperfect information since we do not observe the private information held by the opponent, and therefore do not know which of his action counters we should increment. In this section we will present a new algorithm for this setting. We present it in the context of a
representative motivating game where the opponent is dealt a state of private information and then takes publicly-observable action. 
and present the algorithm for the general setting in Section~\ref{se:general-setting}.   

We are interested in studying the following two-player game setting. Player 1 is given private information state $x_i$ (according to a probability distribution). Then he takes a publicly observable action $a_i$. Player 2 then takes an action after observing player 1's action (but not his private information), and both players receive a payoff. We are interested in player 2's problem of inferring the (assumed stationary) strategy of player 1 after repeated observations of the public action taken (but not the private information). 
Note that this setting is very general. For example, in poker $x_i$ could denote the opponent's private card(s) and $a_i$ the amount bet, and in an ad auction $x_i$ could denote his valuation (e.g., high or low), and $a_i$ could denote the amount he bids~\cite{Tang16:Optimal}.

\subsection{Motivating game and algorithm}
\label{se:motivating-game}
For concreteness and motivation, consider the following poker game instantiation of this setting, where we play the role of player 2. Let's assume that in this two-player game, player 1 is dealt a King (K) and Jack (J) with probability $\frac{1}{2}$, while player 2 is always dealt a Queen. Player 1 is allowed to make a big bet of \$10 (b) or a small bet of \$1 (s), and player 2 is allowed to call or fold. If player 2 folds, then player 1 wins the \$2 pot (for a profit of \$1); if player 1 bets and player 2 calls then the player with the higher card wins the \$2 pot plus the size of the bet. 

\begin{figure}[!ht]
\centering
\includegraphics[scale=0.8]{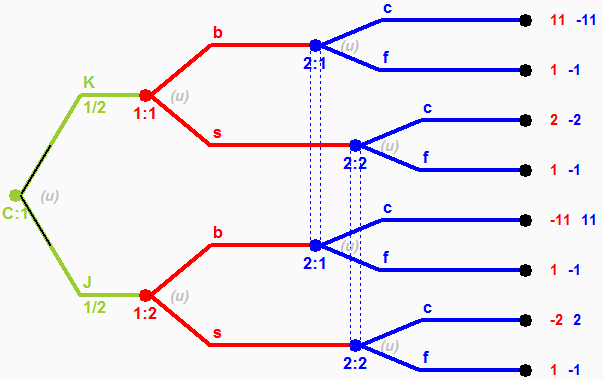}
\caption{Chance deals player 1 king or jack with probability $\frac{1}{2}$ at the green node. Then player 1 selects big or small bet at a red node. Then player 2 chooses call or fold at a blue node.}
\label{fi:game}
\end{figure}

If we observe player 1's card after each hand, then we can apply the approach described above, where we maintain a counter for player 1 choosing each action with each card that is incremented for the selected action. However, if we do not observe player 1's card after the hand (e.g., if we fold), then we would not know whether to increment the counter for the king or the jack. 

To simplify analysis, we will assume that we never observe the opponent's private card after the hand (which is not realistic since we would observe his card if he bets and we call); we can assume that we do not observe our payoff either until all game iterations are complete, since that could allow us to draw inferences about the opponent's card. There are no known algorithms even for the simplified case of fully unobservable opponent's private information. We suspect that an algorithm for the case when the opponent's private information is sometimes observed can be constructed based on our algorithm, and we plan to study this problem in future work.

Let $C$ denote player 1's card and $A$ denote his action. Then $P(C=K) = P(C=J) = \frac{1}{2}$. Let $q_{b|K} \equiv P(A=b|C=K)$ denote the probability that player 1 makes a big bet with king, 
$q_{s|K} \equiv P(A=s|C=K)$ denote the probability that player 1 makes a small bet with king, then $q_{s|K}=1-q_{s|K}$. 
If we were using a Dirichlet prior with parameters $\alpha_1$ and $\alpha_2$, denoted $\mbox{Dir}(\alpha_1,\alpha_2)$, (where $\alpha_1-1$ is the number of times that action $b$ has been observed with a king and $\alpha_2-1$ is the number of times $s$ has been observed with a king), then 
the probability density function is 
$$f_{\mbox{Dir}}(q_{b|K}, q_{s|K};\alpha_1,\alpha_2) 
= \frac{(q_{b|K}^{\alpha_1-1})(q_{s|K})^{\alpha_2-1}}{B(\alpha_1,\alpha_2)} 
= \frac{(q_{b|K}^{\alpha_1-1})(1-q_{b|K})^{\alpha_2-1}}{B(\alpha_1,\alpha_2)}$$
\normalsize
In general given observations O, Bayes' rule gives the following, where $q$ is a mixed strategy that is given mass $p(q)$ under the prior, and $p(q|O)$ is the posterior:

\begin{eqnarray*}
p(q|O) &= &\frac{P(O|q)p(q)}{P(O)} \\
&= &\frac{\sum_{c \in C}P(O,C=c|q)p(q)}{p(O)} \\
&= &\frac{\sum_c P(O|C=c,q)p(c)p(q)}{p(O)}\\
&= &\frac{(P(O|K,q)p(K) + P(O|J,q)p(J))p(q)}{p(O)}\\
&= &\frac{(P(O|K,q) + P(O|J,q))p(q)}{2p(O)}\\
&= &\frac{q_{O|K}p(q) + q_{O|J}p(q)}{2p(O)}
\end{eqnarray*}
\normalsize

Assume action $b$ was observed in a new time step but player 1's card was not. Let $\alpha_{Kb}-1$ be the number of times we observed him play $b$ with $K$ according to the prior, etc. 

\scriptsize
\begin{equation*}
P(q|O) = \frac{q_{b|K}^{\alpha_{Kb}} (1-q_{b|K})^{\alpha_{Ks}-1}q_{b|J}^{\alpha_{Jb}-1}(1-q_{b|J})^{\alpha_{Js}-1}
+q_{b|K}^{\alpha_{Kb}-1} (1-q_{b|K})^{\alpha_{Ks}-1}q_{b|J}^{\alpha_{Jb}}(1-q_{b|J})^{\alpha_{Js}-1}}
{2B(\alpha_{Kb},\alpha_{Ks})B(\alpha_{Jb},\alpha_{Js})p(O)}
\end{equation*}
\normalsize
The general expression for the mean of a continuous random variable is
\begin{eqnarray*}
E[X] &= &\int_x x P(X=x)dx = \int_x x \int_y p(X,Y)dydx \\
&= &\int_x \int_y x p(X,Y)dydx
\end{eqnarray*}
\normalsize
Now we compute the mean of the posterior of the opponent's probability of playing $b$ with J.
\footnotesize
\begin{eqnarray*}
& &P(A = b | O, C=J)  \\
&= &\int_{q_{b|J}} q_{b|J}P(q_{b|J}|O)dq_{b|J}  \\
&= &\frac{\int_{q_{b|J}} \int_{q_{b|K}} q_{b|J} P(q|O) dq_{b|K}dq_{b|J}}{p(O)}\\
&= &\frac{\int \int (q_{b|K}^{\alpha_{Kb}} (1-q_{b|K})^{\alpha_{Ks}-1}q_{b|J}^{\alpha_{Jb}}(1-q_{b|J})^{\alpha_{Js}-1} + q_{b|K}^{\alpha_{Kb}-1} (1-q_{b|K})^{\alpha_{Ks}-1}q_{b|J}^{\alpha_{Jb}+1}(1-q_{b|J})^{\alpha_{Js}-1}) dq_{b|K}dq_{b|J}}
{2B(\alpha_{Kb},\alpha_{Ks})B(\alpha_{Jb},\alpha_{Js})p(O)}\\
&= & \frac{B(\alpha_{Kb}+1,\alpha_{Ks})B(\alpha_{Jb}+1,\alpha_{Js}) + B(\alpha_{Kb}, \alpha_{Ks})B(\alpha_{Jb}+2,\alpha_{Js})}
{2B(\alpha_{Kb},\alpha_{Ks})B(\alpha_{Jb},\alpha_{Js})p(O)}\\
&= &\frac{B(\alpha_{Kb}+1,\alpha_{Ks})B(\alpha_{Jb}+1,\alpha_{Js}) + B(\alpha_{Kb}, \alpha_{Ks})B(\alpha_{Jb}+2,\alpha_{Js})}{Z}
\end{eqnarray*}
\normalsize

The final equation can be obtained by observing that 
$$\int_{q_{b|J}} \int_{q_{b|K}} \frac{q_{b|K}^{\alpha_{Kb}} (1-q_{b|K})^{\alpha_{Ks}-1}q_{b|J}^{\alpha_{Jb}}(1-q_{b|J})^{\alpha_{Js}-1}}{B(\alpha_{Kb}+1,\alpha_{Ks})B(\alpha_{Jb}+1,\alpha_{Js})} dq_{b|K}dq_{b|J}$$
$$= \int_{q_{b|J}} \frac{q_{b|J}^{\alpha_{Jb}}(1-q_{b|J})^{\alpha_{Js}-1}}{B(\alpha_{Jb}+1,b_J)} \int_{q_{b|K}} \frac{q_{b|K}^{\alpha_{Kb}} (1-q_{b|K})^{\alpha_{Ks}-1}}{B(\alpha_{Kb}+1,\alpha_{Ks})} dq_{b|K}dq_{b|J}$$
$$= \int_{q_{b|J}} \frac{q_{b|J}^{\alpha_{Jb}}(1-q_{b|J})^{\alpha_{Js}-1}}{B(\alpha_{Jb}+1,\alpha_{Js})} 1 dq_{b|J} = 1,$$
\normalsize
since the integrands are themselves Dirichlet and all probability distributions integrate to 1. Similarly
$$\int \int \frac{(q_{b|K}^{\alpha_{Kb}-1} (1-q_{b|K})^{\alpha_{Ks}-1}q_{b|J}^{\alpha_{Jb}+1}(1-q_{b|J})^{\alpha_{Js}-1})}{B(\alpha_{Kb}, \alpha_{Ks})B(\alpha_{Jb}+2,\alpha_{Js})} = 1.$$
\normalsize

\normalsize

\normalsize
\normalsize

 Letting $Z$ denote the denominator, we have

 \begin{equation} \label{eq:poker-posterior} P(b|O,J) = \frac{B(\alpha_{Kb}+1,\alpha_{Ks})B(\alpha_{Jb}+1,\alpha_{Js}) + B(\alpha_{Kb}, \alpha_{Ks})B(\alpha_{Jb}+2,\alpha_{Js})}{Z}
 \end{equation}

 where the normalization term $Z$ is equal to
 \begin{eqnarray*}
 &&B(\alpha_{Kb}+1,\alpha_{Ks})B(\alpha_{Jb}+1,\alpha_{Js}) + B(\alpha_{Kb}, \alpha_{Ks})B(\alpha_{Jb}+2,\alpha_{Js}) \\ 
 && + B(\alpha_{Kb}+1,\alpha_{Ks})B(\alpha_{Jb},\alpha_{Js}+1) + B(\alpha_{Kb}, \alpha_{Ks})B(\alpha_{Jb}+1,\alpha_{Js}+1)
 \end{eqnarray*}
\normalsize
$P(s|O,J), P(b|O,K),$ and $P(s|O,K)$ can be computed analogously.  
As stated earlier, $B(\alpha) = \frac{\prod{\Gamma(\alpha_i)}}{\Gamma \left(\sum_i \alpha_i\right)}$ where $\Gamma(n) = (n-1)!$, which can be computed efficiently. 
 
Note that the algorithm we have presented applies for the case where we play one more game iteration and collect one additional observation. However, it is problematic for the general case we are interested in where we play many game iterations, since the posterior distribution  is not Dirichlet, and therefore we cannot just apply the same procedure in the next iteration using the computed posterior as the new prior. We will need to derive a new expression for $P(b|O,J)$ for this setting. Suppose that we have observed the opponent play action $b$ for $\theta_b$ times and $s$ for $\theta_s$ times (in addition to the number of fictitious observations reflected in the prior $\alpha$), though we do not observe his card. Note that there are $\theta_b+1$ possible ways that he could have played $b$ $\theta_b$ times: 0 times with K and $\theta_b$ with J, 1 time with K and $\theta_b-1$ with J, etc. Thus the expression for $p(q|O)$ will have $\theta_b+1$ terms in it instead of two (we can view $\theta_b+1$ as a constant if we assume that the number of game iterations is a constant, but in any case this is linear in the number of iterations).
We study generalization to $n$ private information states and $m$ actions in Section~\ref{se:general-setting}.

We have the new equation
\small
$$P(q|O) = \frac{\sum_{i=0}^{\theta_b} \sum_{j=0}^{\theta_s} q_{b|K}^{\alpha_{Kb}-1+i} (1-q_{b|K})^{\alpha_{Ks}-1+j}q_{b|J}^{\alpha_{Jb}-1+(\theta_b-i)}(1-q_{b|J})^{\alpha_{Js}-1+(\theta_s-j)}}{2B(\alpha_{Kb},\alpha_{Ks})B(\alpha_{Jb},\alpha_{Js})p(O)}$$
\normalsize
Using similar reasoning as above, this gives
\begin{equation}\label{eq:poker-general}  
P(b|O,J) = \frac{\sum_{i=0}^{\theta_b} \sum_{j=0}^{\theta_s} B(\alpha_{Kb}+i,\alpha_{Ks}+j)B(\alpha_{Jb}+\theta_b-i+1,\alpha_{Js}+\theta_s-j)}{Z}
\end{equation}
\normalsize

 The normalization term is 
 \begin{eqnarray*}
 Z &= &\sum_i \sum_j [B(\alpha_{Kb}+i,\alpha_{Ks}+j)B(\alpha_{Jb}+\theta_b-i+1,\alpha_{Js}+\theta_s-j) \\
 &+ &B(\alpha_{Kb}+i,\alpha_{Ks}+j)B(\alpha_{Jb}+\theta_b-i,\alpha_{Js}+\theta_s-j+1)] 
 \end{eqnarray*}
 \normalsize

Thus the algorithm for responding to the opponent is the following. We start with the prior counters on each private information-action combination, $\alpha_{Kb},\alpha_{Ks}$, etc. We keep separate counters $\theta_b, \theta_s$ for the number of times we have observed each action during play. Then we combine these counters according to Equation~\ref{eq:poker-general} in order to compute the strategy for the opponent that is the mean of the posterior given the prior and observations, and we best respond to this strategy, which gives us the same payoff as best responding to the full posterior distribution according to Theorem~\ref{th:mean}. There are only O($n^2$) terms in the expression in Equation~\ref{eq:poker-general}, so this algorithm is efficient.

\subsection{Example}
\label{se:example}
Suppose the prior is that the opponent played b with K 10 times, played s with K 3 times, played b with J 4 times, and played s with J 9 times. Thus $\alpha_{Kb} = 10,\alpha_{Ks} = 3,\alpha_{Jb} = 4,\alpha_{Js}=9.$ Now suppose we observe him play b at the next iteration. Applying our algorithm using Equation~\ref{eq:poker-posterior} gives

$$p(b|O,J) = \frac{B(11,3)B(5,9) + B(10,3)(6,9)}{Z}$$ 
$$= 0.00116550 \cdot 0.00015540 + 0.00151515 \cdot 0.00005550 = \frac{2.65209525e^{-7}}{Z}$$

$$p(s|O,J) = \frac{B(11,3)B(4,10) + B(10,3)(5,10)}{Z}$$ 
$$= 0.00116550 \cdot 0.00034965 + 0.00151515 \cdot 0.00009990 = \frac{5.5888056e^{-7}}{Z}$$


\normalsize

$$\longrightarrow p(b|O,J) = \frac{2.65209525e^{-7}}{2.65209525e^{-7} + 5.5888056e^{-7}} = 0.3218210361.$$
\normalsize

So we think that with a jack he is playing a strategy that bets big with probability 0.322 and small with probability 0.678. Notice that previously we thought his probability of betting big with a jack was $\frac{4}{13} = 0.308$, and had we been in the setting where we always observe his card after gameplay and observed that he had a jack, the posterior probability would be $\frac{5}{14} = 0.357$. 

An alternative ``na\"{i}ve'' (and incorrect) approach would be to increment the counter for $\alpha_{Jb}$ by $\frac{\alpha_{Jb}}{\alpha_{Jb}+\alpha_{Kb}}$, the ratio of the prior probability that he bets big given J to the total prior probability that he bets big. This gives a posterior probability of him betting big with J of $\frac{4+\frac{4}{13}}{14} = 0.308,$ which differs significantly from the correct value. It turns out that this approach is actually equivalent to just using the prior:
$$\frac{x + \frac{x}{x+y}}{x+y+1} \cdot \frac{x+y}{x+y} = \frac{x(x+y) + x}{(x+y+1)(x+y)}
= \frac{x(x+y+1)}{(x+y+1)(x+y)} = \frac{x}{x+y}$$ 
\normalsize

\subsection{Algorithm for general setting}
\label{se:general-setting}
We now consider the general setting where the opponent can have $n$ different states of private information according to an arbitrary distribution $\pi$ and can take $m$ different actions. Assume he is given private information $x_i$ with probability $\pi_i$, for $i = 1, \ldots, n$, and can take action $k_i$, for $i = 1, \ldots, m$. Assume the prior is Dirichlet with parameters $\alpha_{ij}$ for the number of times action $j$ was played with private information $i$ (so the mean of the prior has the player selecting action $k_j$ at state $x_i$ with probability $\frac{\alpha_{ij}}{\sum_j \alpha_{ij}}$). 

\small
$$P(C=x_i) = \pi_i; $$ 
$$P(A=k_j|C=x_i) = q_{k_j|x_i};$$ 
$$q_{k_j|x_i} \sim f_{\mbox{Dir}}(q_{k_j|x_i};\alpha_{i1},\ldots,a_{im}) = \frac{\prod_j q^{\alpha_{ij}-1}_{{k_j}|x_i}}
{B(\alpha_{i1},\ldots,\alpha_{im})}$$
\normalsize

As before, using Bayes' rule we have

\begin{eqnarray*}
p(q|O) &= &\frac{P(O|q)p(q)}{P(O)} \\
&= &\frac{\sum_i P(O,C=x_i|q)p(q)}{p(O)} \\ 
&= &\frac{\sum_i P(O|C=x_i,q) \pi_i p(q)}{p(O)} \\ 
&= &\frac{p(q) \sum_i P(O|x_i,q) \pi_i}{p(O)} \\
&= &\frac{\sum_i P(O|x_i) p(q) \pi_i }{p(O)}
\end{eqnarray*}

Now assume that action $k_{j^*}$ was observed in a new time step, while the opponent's private information was not observed. 

$$P(q|O) = \frac{\sum_{i=1}^n \left[\pi_i q_{k_{j^*}|x_i} \prod_{h=1}^m \prod_{j=1}^n q_{k_h|x_j}^{\alpha_{jh}-1} \right]}
{p(O)\prod_{i=1}^n B(\alpha_{i1},\ldots,\alpha_{im})}$$

We now compute the expectation for the posterior probability that the opponent plays $k_{j^*}$ with private information $x_{i^*}$ as done in Section~\ref{se:motivating-game}.

\begin{eqnarray*}
P(A= k_{j^*} | O, C = x_{i^*}) &= &
\frac{\int \left[ q_{k_j^*|x_i^*} \sum_{i=1}^n \left[\pi_i q_{k_{j^*}|x_i} \prod_{h=1}^m \prod_{j=1}^n q_{k_h|x_j}^{\alpha_{jh}-1}\right] \right]}{p(O)\prod_{i=1}^n B(\alpha_{i1},\ldots,\alpha_{im})} \\
&= &\frac{\sum_i \left[ \pi_i \prod_j B(\gamma_{1j},\ldots,\gamma_{nj}) \right]} {Z},
\end{eqnarray*}
where $\gamma_{ij} = \alpha_{ij}+2$ if $i = i^*$ and $j = j^*$, $\gamma_{ij} = \alpha_{ij}+1$ if  $j = j^*$ and $i \neq i^*$, and $\gamma_{ij} = \alpha_{ij}$ otherwise. If we denote the numerator by $\tau_{i^*j^*}$ then $Z = \sum_{i^*} \tau_{i^*j^*}.$ Notice that the product is over $n$ terms, and therefore the total number of terms will be exponential in $n$ (it is O($m \cdot 2^n$)).

For the case of multiple observed actions, the posterior is not Dirichlet and cannot be used directly as the prior for the next iteration. Suppose we have observed action $k_j$ for $\theta_j$ times (in addition to the number of fictitious times indicated by the prior counts $\alpha_{ij}$). We compute $P(q|O)$ analogously as

$$P(q|O) = \frac{\sum_{i=1}^n \left[\pi_i \sum_{\{\rho_{ab}\}} \prod_{h=1}^m \prod_{j=1}^n q_{k_h|x_j}^{\alpha_{jh}-1+\rho_{jh}} \right]}
{p(O)\prod_{i=1}^n B(\alpha_{i1},\ldots,\alpha_{im})},$$
where the $\sum_{\{\rho_{ab}\}}$ is over all values $0 \leq \rho_{ab} \leq \theta_b$ with $\sum_a \rho_{ab} = \theta_b$ for each $b$, for $1 \leq a \leq n$, $1 \leq b \leq m$. We can write this as 
$$\sum_{\{\rho_{ab}\}} = \sum_{\rho_{1b} = 0}^{\theta_b} \sum_{\rho_{2b} = 0}^{\theta_b - \rho_{1b}} \ldots  
\sum_{\rho_{n-1,b} = 0}^{\theta_b - \sum_{r = 0}^{n-2} \rho_{rb}}
\sum_{\rho_{nb} = \theta_b - \sum_{r = 0}^{n-2} \rho_{rb}}^{\theta_b - \sum_{r = 0}^{n-1} \rho_{rb}}.$$
\normalsize



Here, the expression for the full posterior distribution $P(q|O)$ is 
\begin{equation*}
\label{eq:general}
P(q|O) = 
\frac{\sum_i \left[ \pi_i \sum_{\{\rho_{ab}\}} \prod_h B(\alpha_{1h}+\rho_{1h},\ldots,\alpha_{nh}+\rho_{nh})\right]} {Z}
\end{equation*}
(Note that we could marginalize as before to compute the mean of the posterior strategy for arbitrary indices $\hat{i}, \hat{j}$, but we omit these details to avoid making the formula unnecessarily complicated.)

For each $b$, the number of terms equals the number of ways of distributing the $\theta_b$ observations amongst the $n$ possible private information states, which equals 
$$C(\theta_b + n -1, n-1) = \frac{(\theta_b +n-1)!}{(n-1)!\theta_b!} $$
\normalsize
Hence the total number of terms in the summation is upper bounded by 
$O\left( \frac{(T +n)!}{n!T!} \right)$, where $T$ is the total number of game iterations which upper bounds the $\theta_b$'s. Since we must do this for each of the $m$ actions, the total number of terms is $O\left( \left(\frac{(T +n)!}{n!T!} \right)^m \right).$ So the number of terms is exponential in the number of private information states and actions, but polynomial in the number of iterations.

In the following section, we will look at how to compute the product of beta functions for each term by approximating the product with an exponential of the sum of terms, improving the computation complexity for both the single and multiple observation settings.



\subsection{Calculation for product of beta distributions}
\label{se:beta}
From the previous section, we know that the posterior probability that the opponent plays $k_{j^*}$ with private information $x_{i^*}$ is
\begin{eqnarray*}
P(A= k_{j^*} | O, C = x_{i^*}) 
&= &\frac{\sum_i \left[ \pi_i \prod_j B(\gamma_{1j},\ldots,\gamma_{nj}) \right]} {Z},
\end{eqnarray*}
One computational bottleneck for computing the posterior probability is to compute the product of beta distributions.

Now we study the product of beta distributions $\prod_j B(\gamma_{1j},\ldots,\gamma_{nj})$.  We will derive an analytical formula so that evaluation of this formula is computationally feasible.

We prove the following theorem:
\begin{theorem}
Define $\gamma_{j}=\sum_{i=1}^n \gamma_{ij}$, then define the empirical probability distribution $\hat P_j(i)= \frac{\gamma_{ij}}{\sum_{i=1}^n \gamma_{ij}} = \frac{\gamma_{ij}}{\gamma_{j}} $. Define the Gamma function $\Gamma(x)= \int_0^\infty x^{z-1} e^{-x}\, \mathrm{d}x$, for integer $x$, $\Gamma(x)= (x-1)!$. Now define the entropy of $\hat P_i$ as $E(\hat P_j)= -\sum_{i=1}^n \hat P_j(i) \ln\hat P_j(i).$ Then we have
\begin{eqnarray*}
 \prod_{j=1}^m B(\gamma_{1j},\ldots,\gamma_{nj}) = 
 \exp \left( \sum_{j=1}^m \left(-\gamma_{j}E(\hat P_j)-\frac{1}{2}(n-1)\ln(\gamma_{j})+\sum_{i=1}^n \ln(P_j(i))+d \right) \right).
\end{eqnarray*}
Here $d$ is a constant such that $\tfrac12\ln(2\pi) n-1\leq d\leq n-\tfrac12\ln(2\pi)$, where $\ln(2\pi) \approx 0.92$.
\end{theorem}
\begin{proof}
We know from the definition of beta function that 
\begin{eqnarray*}
B(\gamma_{1j},\ldots,\gamma_{nj}) 
 = 
\frac{\prod_{i=1}^n \Gamma(\gamma_{ij})}{\Gamma \left(\sum_{i=1}^n \gamma_{ij} \right)}
\end{eqnarray*}
Let's use a version of Stirling's formula with bounds valid for all positive integers $z$, 
\begin{eqnarray*}
 \sqrt{2\pi z}\left(\frac{z}{e}\right)^z\leq \Gamma(z+1)= z!\leq   e\sqrt{z}\left(\frac{z}{e}\right)^z.
\end{eqnarray*}
Therefore, $\Gamma(z+1)=C(z) \sqrt{z}\left(\frac{z}{e}\right)^z$ and $C(z)$ is between $ \sqrt{2\pi}=2.5066$ and $e=2.71828$.
And since we want to express the product of Gamma function, we consider
\begin{eqnarray*}
\ln\Gamma(z)=  z\ln z - z - \tfrac12\ln(z)+c(z)
\end{eqnarray*}
where $c(z)$ is between $\tfrac12\ln(2\pi)\approx 0.92$ and $1$.

Now we can look at the product of beta functions,
\begin{eqnarray*}
 \prod_{j=1}^m B(\gamma_{1j},\ldots,\gamma_{nj}) 
 &=& 
 \prod_{j=1}^m \left(\frac{\prod_{i=1}^n \Gamma(\gamma_{ij})}{\Gamma \left(\sum_{i=1}^n \gamma_{ij} \right)} \right) 
 \\&=& 
 \exp \left(\sum_{j=1}^m \left(\sum_{i=1}^n\ln\Gamma(\gamma_{ij}) -\ln \Gamma \left(\sum_{i=1}^n \gamma_{ij} \right)\right)\right) 
\end{eqnarray*}
Now we want to use Stirling's formula with bounds and the definition of entropy to reduce the terms in the exponential.  

\small
\begin{eqnarray*}
& & \sum_{j=1}^m \left( \sum_{i=1}^n\ln\Gamma(\gamma_{ij}) -\ln \left( \Gamma \left(\sum_{i=1}^n \gamma_{ij} \right) \right) \right)\\
  &=&
\sum_{j=1}^m \left( \sum_{i=1}^n \left(\gamma_{ij}\ln \gamma_{ij} - \gamma_{ij} - \frac{\ln(\gamma_{ij})}{2}+c(\gamma_{ij}) \right) -
\left( \sum_{i=1}^n\gamma_{ij}\ln \sum_{i=1}^n\gamma_{ij} - \sum_{i=1}^n\gamma_{ij} - \frac{\ln \left(\sum_i \gamma_{ij} \right)}{2} +c \left(\sum_{i=1}^n\gamma_{ij} \right) \right)\right)
\\ &=& 
\sum_{j=1}^m \left( \sum_{i=1}^n\gamma_{ij}\sum_{i=1}^n \frac{\gamma_{ij}}{\sum_{i=1}^n \gamma_{ij}}
\ln \left(\frac{\gamma_{ij}}{\sum_i \gamma_{ij}} \right)-\frac{n-1}{2}\ln \left( \sum_{i=1}^n \gamma_{ij}\right)+\sum_{i=1}^n \ln \left(\sum_{i=1}^n \frac{\gamma_{ij}}{\sum_i \gamma_{ij}} \right)+ \sum_{i=1}^n c(\gamma_{ij})-c\sum_{i=1}^n\gamma_{ij}\right)
\\ &=& 
\sum_{j=1}^m \left(  \left(-\gamma_{j}E(\hat P_j) \right) -\frac{1}{2}(n-1)\ln(\gamma_{j})+\sum_{i=1}^n \ln(P_j(i))+
\sum_{i=1}^n c(\gamma_{ij})-c \left(\sum_{i=1}^n\gamma_{ij}\right)\right)
 \\ &=& 
\sum_{j=1}^m \left(  \left(-\gamma_{j}E(\hat P_j) \right)-\frac{1}{2}(n-1)\ln(\gamma_{j})+\sum_{i=1}^n \ln(P_j(i))+d\right).
\end{eqnarray*}
\normalsize
Here $n\tfrac12\ln(2\pi)-1\leq d\leq n-\tfrac12\ln(2\pi)$.

The first equation is Stirling's formula with bounds.
The second expands and reorganizes the terms.
The third equation applies the definition $\gamma_{j}=\sum_{i=1}^n \gamma_{ij}$, the definition of empirical probability distribution 
$\hat P_j(i)= \frac{\gamma_{ij}}{\sum_{i=1}^n \gamma_{ij}}, $
and the definition of entropy for an empirical probability distribution 
$$E(\hat P_j)= -\sum_{i=1}^n \hat P_j(i) \ln(\hat P_j(i)) = -\sum_{i=1}^n \frac{\gamma_{ij}}{\sum_{i=1}^n \gamma_{ij}} 
\ln \left( \frac{\gamma_{ij}}{\sum_{i=1}^n \gamma_{ij}} \right).$$
The fourth equation uses the bound of $c(\gamma_{ij})$ between $\tfrac12\ln(2\pi)$ and $1$ on each constant, to get $\tfrac12\ln(2\pi) n -1\leq d\leq n-\tfrac12\ln(2\pi)$.

Therefore, we have  
\begin{eqnarray*}
& & \prod_{j=1}^m B(\gamma_{1j},\ldots,\gamma_{nj}) 
\\&=& \exp \left( \sum_{j=1}^m \left( \sum_{i=1}^n\ln\Gamma(\gamma_{ij}) -\ln \Gamma ( \sum_{i=1}^n \gamma_{ij}  ) \right) \right) 
\\&=& \exp\left(\sum_{j=1}^m \left(\left(-\gamma_{j}E(\hat P_j) \right)-\frac{1}{2}(n-1)\ln(\gamma_{j})+\sum_{i=1}^n \ln(P_j(i))+d \right) \right)
\end{eqnarray*}
\end{proof}

We remark that the approximation is actually pretty tight numerically, since in Stirling's formula, we get a constant $c(z)$ between $\tfrac12\ln(2\pi)\approx 0.92$ and $1$. And in numerical computation we can just take $0.92$. It would be especially accurate when these $\gamma_{ij}$ are large, since the lower bound $\tfrac12\ln(2\pi)\approx 0.92$ is accurate in the asymptotic version of Stirling's formula.

Using this formula, we can evaluate $P(A= k_{j^*} | O, C = x_{i^*})$ efficiently. For a single observation,
\begin{eqnarray*}
P(A= k_{j^*} | O, C = x_{i^*}) 
&= &\frac{\sum_i \left[ \pi_i \prod_j B(\gamma_{1j},\ldots,\gamma_{nj}) \right]} {Z},\\
&= &\frac{\sum_i \left[ \pi_i \exp\left(\sum_{j=1}^m \left(-\gamma_{j}E(\hat P_j)-\frac{1}{2}(n-1)\ln(\gamma_{j})+\sum_{i=1}^n \ln(P_j(i))+d \right)\right) \right]} {Z},
\end{eqnarray*}
where $\gamma_{ij} = \alpha_{ij}+2$ if $i = i^*$ and $j = j^*$, $\gamma_{ij} = \alpha_{ij}+1$ if  $j = j^*$ and $i \neq i^*$, and $\gamma_{ij} = \alpha_{ij}$ otherwise. And the constant $d$ is chosen by $\tfrac12\ln(2\pi) n -1\leq d\leq n-\tfrac12\ln(2\pi)$. If we denote the numerator by $\tau_{i^*j^*}$ then $Z = \sum_{i^*} \tau_{i^*j^*}.$

The complexity for computing this posterior is $n$ times the complexity of computing the terms inside the exponential. The complexity of that is $mn$ since the computation for entropy is $n$, and we sum over $m$ in the outside.
Therefore, the complexity for computing this posterior is $n^2 m$.

Similarly, in the multiple observation case,  
\small
\begin{eqnarray*}
P(q | O) 
&=&\frac{\sum_i \left[ \pi_i \sum_{\{\rho\}} \prod_h B(\gamma_{1j},\ldots,\gamma_{nj})\right]} {Z}\\
&=& \frac{\sum_i \left[ \pi_i \sum_{\{\rho\}} 
\exp\left(\sum_{j=1}^m \left(-\gamma_{j}E(\hat P_j)-\frac{1}{2}(n-1)\ln(\gamma_{j})+\sum_{i=1}^n \ln(P_j(i))+d \right) \right)
\right]} {Z}, 
\end{eqnarray*}
\normalsize
where $\gamma_{ij} = \alpha_{ij}+\rho_{ij}.$
 And $\sum_{\{\rho\}}$ is over all values $0 \leq \rho_{ab} \leq \theta_b$ with $\sum_a \rho_{ab} = \theta_b$ for each $b$, for $1 \leq a \leq n$, $1 \leq b \leq m$, as 
$$\sum_{\{\rho\}} = \sum_{\rho_{1b} = 0}^{\theta_b} \sum_{\rho_{2b} = 0}^{\theta_b - \rho_{1b}} \ldots \sum_{\rho_{nb} = 0}^{\theta_b - \sum_{r = 0}^{n-1} \rho_{rb}}.$$

The complexity for computing this posterior is  $n^2 m \cdot C(\theta_b + n -1, n-1)$, as the product of the complexity in the previous case $n^2 m$ and the number of ways to write $\theta_b$ as finite sums of observations of at most $n$ buckets, which is $C(\theta_b + n -1, n-1)$. If the upper bound on $\theta_b$ is $T$, then as we discussed above, the total number of terms in the summation is upper bounded by $O\left( \frac{(T +n)!}{n!T!} \right)$, where $T$ is the total number of game iterations which upper bounds the $\theta_b$'s. And the  complexity of computing this posterior is  $ O\left( \frac{(T +n)!}{n!T!} n^2 m \right).$ Thus, overall this approach for computing products of the beta function leads to an exponential improvement in the running time of the algorithm for one observation, and reduces the dependence on $m$ for the multiple observation setting from exponential to linear, though the complexity still remains exponential in $n$ and $T$ for the latter.

\section{Algorithm for uniform prior distribution}
\label{se:uniform}
Another natural prior that has been studied previously is the uniform distribution over a polyhedron. This can model the situation when we think the opponent is playing uniformly at random within some region of a fixed strategy, such as a specific Nash equilibrium or a ``population mean'' strategy based on historical data. Prior work has used this model to generate a class of opponents who are significantly more sophisticated than opponents who play uniformly at random over the entire space~\cite{Ganzfried15:Safe}). For example, in rock-paper-scissors, we may think the opponent is playing a strategy uniformly at random out of strategies that play each action with probability within [0.31,0.35], as opposed to completely random over [0,1]. 

Let $v_{i,j}$ denote the $j$th vertex for player $i$, where vertices correspond to mixed strategies. Let $p^0$ denote the prior distribution over vertices, where $p^0_{i,j}$ is the probability that player $i$ plays the strategy corresponding to vertex $v_{i,j}$. Let $V_i$ denote the number of vertices for player $i$. Algorithm~\ref{al:uniform} computes the Bayesian best response in this setting. Correctness follows straightforwardly by applying Corollary~\ref{co:posterior-mean} with the formula for the mean of the uniform distribution.

\begin{algorithm}[!ht]
\caption{Algorithm for opponent exploitation with uniform prior distribution over polyhedron \normalsize}
\label{al:uniform} 
\textbf{Inputs}: Prior distribution over vertices $p^0$, response functions $r_t$ for $0 \leq t \leq T$
\begin{algorithmic}
\State $M_0 \gets$ strategy profile assuming opponent $i$ plays each vertex $v_{i,j}$ with probability $ p^0_{i,j} = \frac{1}{V_i}$ 
\State $R_0 \gets r_0(M_0)$
\State Play according to $R_0$
\For {$t = 1$ to $T$}
\For {$i = 1$ to $N$}
\State $a_i \gets $ action taken by player $i$ at time step $t$
\For {$j = 1$ to $V_i$}
\State $p^t_{i,j} \gets p^{t-1}_{i,j} \cdot v_{i,j}(a_i)$ 
\EndFor
\State Normalize the $p^t_{i,j}$'s so they sum to 1
\EndFor
\State $M_t \gets$ strategy profile assuming opponent $i$ plays each vertex $v_{i,j}$ with probability $p^t_{i,j}$ 
\State $R_t \gets r_t(M_t)$
\State Play according to $R_t$
\EndFor
\end{algorithmic}
\end{algorithm}
\normalsize

\section{Experiments}
\label{se:exp}
We ran experiments on the game described in Section~\ref{se:motivating-game}. For the beta function computations we used the Colt Java math library.\footnote{\url{https://dst.lbl.gov/ACSSoftware/colt/}} 
For our first set of experiments we tested our basic algorithm which assumes that we observe a single opponent action (Equation~\ref{eq:poker-posterior}). We varied the Dirichlet prior parameters to be uniform in \{1,n\} to explore the runtime as a function of the size of the prior (since computing larger values of the eta function can be challenging). The results (Table~\ref{ta:basic}) show that the computation is very fast even for large $n$, with running time under 8 microseconds for $n = 500$. However, we also observe frequent numerical instability for large $n$. The second row shows the percentage of the trials for which the algorithm produced a result of ``NaN'' (which typically results from dividing zero by zero). This jumps from 0\% for $n = 50$ to 8.8\% for $n = 100$ to 86.9\% for $n = 200$. This is due to instability of algorithms for computing the beta function. We used to our knowledge the best publicly available beta function solver, but perhaps there could be a different solver that leads to better performance in our setting (e.g., it trades off runtime for additional precision). In any case, despite the cases of instability, the results indicate that the algorithm runs extremely fast for hundreds of prior observations, and since it is exact, it is the best algorithm for the settings in which it produces a valid output. Note that $n = 100$ corresponds to 400 prior observations on average since there are four parameters, and that the experiments in previous work used a horizon of 200 hands played in a match against an opponent~\cite{Southey05:Bayes}.

\begin{table}[!ht]
\centering
\begin{tabular}{|*{7}{c|}} \hline
$n$ &10 &20 &50 &100 &200 &500\\ \hline
Time &0.0005 &0.0008 &0.0018 &0.0025 &0.0034 &0.0076\\ \hline
NaN &0 &0 &0 &0.0883 &0.8694 &0.9966\\ \hline
\end{tabular}
\caption{Results of modifying Dirichlet parameters to be U\{1,n\} over one million samples. First row is average runtime in milliseconds. Second row is percentage of the trials that output ``NaN.'' \normalsize}
\label{ta:basic}
\end{table}
\normalsize

We tested our generalized algorithm from Equation~\ref{eq:poker-general} for different numbers of observations, keeping the prior fixed. We used a Dirichlet prior with all parameters equal to 2 as has been done in prior work~\cite{Southey05:Bayes}. We observe (Table~\ref{ta:multi}) that the algorithm runs quickly for large numbers of observations, though again it runs into numerical instability for large values. As one example, the algorithm runs in 19 milliseconds for $\theta_b = 101$, $\theta_s = 100$. 

\begin{table}[!ht]
\centering
\begin{tabular}{|*{8}{c|}} \hline
$n$ &10 &20 &50 &100 &200 &500 &1000\\ \hline
Time &0.015 &0.03 &0.36 &2.101 &10.306 &128.165 &728.383\\ \hline
NaN &0 &0 &0 &0 &0.290 &0.880 &0.971\\ \hline
\end{tabular}
\caption{Results using Dirichlet prior with all parameters equal to 2 and $\theta_b$, $\theta_s$ in U\{1,n\} averaged over 1,000 samples. First row is average runtime in milliseconds, second row is percentage of trials that produced ``NaN.'' 
\normalsize}
\label{ta:multi}
\end{table}
\normalsize

We compared our algorithm against the three heuristics described in previous work~\cite{Southey05:Bayes}. The first heuristic Bayesian Best Response (BBR) approximates the opponent's strategy by sampling strategies according to the prior and computing the mean of the posterior over these samples, then best-responding to this mean strategy. Their Max A Posteriori Response heuristic (MAP) samples strategies from the prior, computes the posterior value for these strategies, and plays a best response to the one with highest posterior value. Thompson's Response samples strategies from the prior, computes the posterior values, then samples one strategy for the opponent from these posteriors and plays a best response to it. For all approaches we used a Dirichlet prior with the standard values of 2 for all parameters. For all the sampling approaches we sampled 1,000 strategies from the prior for each opponent and used these strategies for all hands against that opponent (as was done in prior work~\cite{Southey05:Bayes}). Note that one can draw samples $x_1,\ldots,x_K$ from a Dirichlet distribution by first drawing $K$ independent samples $y_1,\ldots,y_K$ from Gamma distributions each with density 
$\mbox{Gamma}(\alpha_i,1) = \frac{y^{\alpha_i - 1}_i e^{-y_i}}{\Gamma(\alpha_i)}$
and then setting 
$x_i = \frac{y_i}{\sum_j y_j}.$\footnote{\url{https://en.wikipedia.org/wiki/Dirichlet_distribution#Random_number_generation}}

We also compared against the payoff of a full best response strategy that knows the actual mixed strategy of the opponent, not just a distribution over his strategies, as well as the Nash equilibrium strategy.\footnote{Note that the Nash equilibrium for player 2 is to call a big bet with probability $\frac{1}{4}$ and a small bet with probability 1 (the equilibrium for player 1 is to always bet big with a king and to bet big with probability $\frac{5}{6}$ with a jack).} Note that the game has a value to us (player 2) of -0.75, so negative values are not necessarily indicative of ``losing.'' 

Table~\ref{ta:comparison} shows that our exact Bayesian best response algorithm (EBBR) outperforms the heuristic approaches, as expected since it is optimal when the opponent's strategy is drawn from the prior (though performance is very similar to BBR and not statistically distinguishable until 25 iterations). BBR performed best out of the sampling approaches, which is not surprising because it is trying to approximate the optimal approach while the others are optimizing a different objective. All of the sampling approaches outperformed just following the Nash equilibrium, and as expected all exploitation approaches performed worse than playing a best response to the opponent's actual strategy. Note that, against an opponent drawn from a Dirichlet distribution with all parameters equal to 2 and no further observations of his play, our best response would be to always call, which gives us expected payoff of zero. Thus for the initial column the actual value for EBBR when averaged over all opponents would be zero. Against this distribution the Nash equilibrium has expected payoff $-0.375$. 

\begin{table}[!ht]
\centering
\begin{tabular}{|*{4}{c|}} \hline
Algorithm & Initial & 10 & 25\\ \hline
\textbf{EBBR} & \boldmath$-0.00003 \pm 0.0003$ &\boldmath$-0.0004 \pm 0.0009$ &\boldmath$0.0002 \pm 0.0008$\\ \hline
BBR & $-0.00003 \pm 0.0003$ &$-0.0004 \pm 0.0009$ &$-0.0065 \pm 0.0008$\\ \hline
MAP & $-0.1649 \pm 0.0002$ &$-0.2025 \pm 0.0007$ &$-0.2664 \pm 0.0007$\\ \hline
Thompson &$-0.2098 \pm 0.0002$ &$-0.2224 \pm 0.0007$ &$-0.2996 \pm 0.0007$\\ \hline
FullBR &$0.4975 \pm 0.0002$ &$0.4971 \pm 0.0006$ &$0.4978 \pm 0.0005$\\ \hline
Nash &$-0.3750 \pm 0.0000$ &$-0.3749 \pm 0.0001$ &$-0.3751 \pm 0.0001$\\ \hline
\end{tabular}
\caption{Comparison with algorithms from prior work, full best response, and Nash equilibrium using Dirichlet prior with parameters equal to 2. Sampling algorithms use 1000 samples. For initial column we sampled 100 million opponents from the prior, for 10 rounds we sampled one million, and for 25 rounds 500,000. Results are average winrate per hand over all opponents with 95\% confidence intervals.}
\label{ta:comparison}
\end{table}

On the positive side the exploitation approaches (particularly EBBR and BBR) are able to significantly outperform the Nash equilibrium strategy when given access to a reliable prior distribution; however, none of them are able to improve over time as a result of additional observations (EBBR and BBR perform around the same with more observations while Thompson and MAP perform noticeably worse). This indicates that, for this setting at least, just observing the opponent's public action and not private information is not additionally useful in comparison to the performance variance and the noise introduced from sampling. In order to successfully learn beyond the prior in imperfect-information settings, algorithms will need access to some of the opponents' private information. Previous experiments had also shown that when the sampling approaches are played against opponents drawn from the prior, the winning rates converge, typically very quickly (even with access to the opponent's private information in certain hands that went to \emph{showdown}): ``The independent Dirichlet prior is very broad, admitting a wide variety of opponents. It is encouraging that the Bayesian approach is able to exploit even this weak information to achieve a better result.''~\cite{Southey05:Bayes} 

We also tested the effect of using only 10 samples of the opponent's strategy for the sampling approaches. The approaches would then have a noisier estimate of the opponent's strategy and should achieve lower performance against the actual strategy, though run significantly faster.

\begin{table}[!ht]
\centering
\begin{tabular}{|*{5}{c|}} \hline
Alg & Initial & 10 & 25 &100\\ \hline
\textbf{EBBR} & \boldmath$.0001 \pm .0003$ &\boldmath-$.0003 \pm .0003$ &\boldmath$.0002 \pm .0002$ &\boldmath-$.0014 \pm .0005$\\ \hline
BBR &-$.0662 \pm .0003$ &-$.0902 \pm .0003$ &-$.1634 \pm .0002$ &-$.3127 \pm .0004$\\ \hline
MAP &-$.1699 \pm .0002$ &-$.2060 \pm .0002$ &-$.2657 \pm .0001$ &-$.3082 \pm .0004$\\ \hline
Thompson &-$.2118 \pm .0002$ &-$.2247 \pm .0002$ &-$.2844 \pm .0001$ &-$.3725 \pm .0004$\\ \hline
FullBR &$.4976 \pm .0002$ &$.4973 \pm .0002$ &$.4975 \pm .0001$ &$.4969 \pm .0003$\\ \hline
Nash &-$.3750 \pm .0000$ &-$.3750 \pm .0000$ &-$.3750 \pm .0000$ &-$.3750 \pm .0001$\\ \hline
\end{tabular}
\caption{Comparison of our algorithm with algorithms from prior work (BBR, MAP, Thompson), full best response, and Nash equilibrium using Dirichlet prior with parameters equal to 2. The sampling algorithms each use 10 samples from the opponent's strategy (as opposed to 1000 samples from our earlier analysis). For the initial column we sampled 100 million opponents from the prior, for 10 and 25 rounds we sampled ten million, and 300,000 for 100 rounds.}  
\label{ta:comparison-10samples}
\end{table}
\normalsize

Thompson and MAP performed very similarly using 10 vs. 1000 samples (these approaches essentially end up selecting a single strategy from the set of samples to be used as the model, and the results indicate that they are relatively insensitive to the number of samples used), but BBR performs significantly worse. 
While the performance between EBBR and BBR was statistically indistinguishable for 1000 samples, EBBR significantly outperforms BBR with 10 samples, particularly for more iterations.
As before the sampling approaches seem to actually perform worse over time as the noise propagates, while the performance of EBBR remains about the same. The dropoff of BBR is particularly significant. The results indicate that EBBR would be particularly preferable over the sampling approaches if the number of available samples is small (e.g., due to running time considerations) and as the number of game iterations increases (though eventually EBBR can run into numerical stability issues described earlier).

\section{Conclusion}
\label{se:conclusion}

One of the most fundamental problems in game theory is learning to play optimally against opponents who may make mistakes. We presented the first exact algorithm for performing exploitation in imperfect-information games in the Bayesian setting using the most well-studied prior distribution for this problem, the Dirichlet distribution. 
Previously an exact algorithm had only been presented for normal-form games, and the best previous algorithm was a heuristic with no guarantees. 
We demonstrated experimentally that our algorithm can be practical and that it outperforms the best prior approaches, though it can run into numerical stability issues for large numbers of observations. 

We presented a general meta-algorithm and new theoretical framework for studying opponent exploitation. Future work can extend our analysis to many important settings. For example, we would like to study the setting when the opponent's private information is only sometimes observed (we expect our approach can be extended easily to this setting) and general sequential games where the agents can take multiple actions (which we expect to be hard, as indicated by the analysis in the tech report). We would also like to extend analysis for any number of agents. Our algorithm is not specialized for two-player zero-sum games (it applies to general-sum games); if we are able to compute the mean of the posterior strategy against multiple opponent agents, then best responding to this strategy profile is just a single agent optimization and can be done in time linear in the size of the game regardless of the number of opponents. While the Dirichlet is the most natural prior for this problem, we would also like to study other important distributions. We presented an algorithm for the uniform prior distribution over a polyhedron, which could model the situation where we think the opponent is playing a strategy from a uniform distribution in a region around a particular strategy, such as a specific equilibrium or a ``population mean'' based on historical data.

Opponent exploitation is a fundamental problem, and our algorithm and extensions could be applicable to many domains that are modeled as an imperfect-information games. For example, many security game models have imperfect information, e.g.,~\cite{Letchford10:Computing,Kiekintveld10:Robust}, and opponent exploitation in security games has been a very active area of study, e.g.,~\cite{Pita10:Robust,Nguyen13:Analyzing}. It has also been proposed recently that opponent exploitation can be important in medical treatment~\cite{Sandholm15:Steering}.

\bibliographystyle{ACM-Reference-Format}
\bibliography{D://FromBackup/Research/refs/dairefs}

\end{document}